\documentclass[]{interact}

\numberwithin{equation}{section}

\usepackage{hyperref}
\usepackage{xcolor}
\usepackage{epstopdf}
\usepackage[caption=false]{subfig}

\usepackage{natbib}
\bibpunct[, ]{(}{)}{;}{a}{}{,}
\renewcommand\bibfont{\fontsize{10}{12}\selectfont}

\theoremstyle{plain}
\newtheorem{theorem}{Theorem}[section]
\newtheorem{lemma}[theorem]{Lemma}

\theoremstyle{definition}

\theoremstyle{remark}

\newcommand{\Prob}{\mathsf{P}}
\newcommand{\Expect}{\mathsf{E}}
\DeclareMathOperator*{\esssup}{ess\,sup}

\begin{document}


\title{Robust Quickest Correlation Change Detection in High-Dimensional Random Vectors}

\author{
	\name{Assma Alghamdi, Taposh Banerjee\textsuperscript{a}\thanks{CONTACT: Taposh Banerjee. Email: taposh.banerjee@pitt.edu}, and Jayant Rajgopal}
	\affil{Department of Industrial Engineering, University of Pittsburgh}
}
\maketitle

\begin{abstract}
Detecting changes in high-dimensional vectors presents significant challenges, especially when the post-change distribution is unknown and time-varying. This paper introduces a novel robust algorithm for correlation change detection in high-dimensional data. The approach utilizes the summary statistic of the maximum magnitude correlation coefficient to detect the change. This summary statistic captures the level of correlation present in the data but also has an asymptotic density. The robust test is designed using the asymptotic density. The proposed approach is robust because it can help detect a change in correlation level from some known level to unknown, time-varying levels. The proposed test is also computationally efficient and valid for a broad class of data distributions. The effectiveness of the proposed algorithm is demonstrated on simulated data.  
\end{abstract}


\section{Introduction} \label{s:intro}
In data-driven decision-making, the real-time processing and analysis of datasets is crucial across various industries, from finance and engineering to healthcare and environmental monitoring. Given  the wealth of information available today, an important challenge is that of effectively detecting changes in underlying processes, particularly in scenarios where timely action is imperative. Recognizing these changes is paramount, and failure to do so could potentially have dire consequences, ranging from industrial disasters to healthcare crises. For example, effective monitoring and maintenance of industrial or production processes often rely on Statistical Process Control (SPC) charts, which are essential tools for detecting changes and facilitating timely corrective actions. To enhance these capabilities, SPC algorithms have been developed to optimize sample rates and sample sizes, thereby reducing the costs associated with sampling while maintaining effective monitoring (\cite{tagaras1998survey}). These advancements underscore the potential of process control charts and other SPC methods to improve monitoring efficiency in high-dimensional data environments. Section~\ref{sec:examples} below provides a few more motivating examples of change detection.  

In statistics, the problem of real-time change detection is studied under the topic of quickest change detection (\cite{veeravalli2014quickest, poor-hadj-qcd-book-2009, tart-niki-bass-2014, tart-book-2019, basseville1993detection, brodsky1993nonparametric}). When the data distributions before and after a change are known, several optimal algorithms have been developed to identify the changes. These algorithms employ a stopping rule, where a sequence of statistics based on the likelihood ratio of observations is calculated. A change is declared when an element in this sequence exceeds a predetermined threshold, which is set to control the rate of false alarms. Among the most notable algorithms are the Shiryaev algorithm (\cite{shiryaev1963}), which was later extended by  \cite{tartakovsky2005}, and the Cumulative Sum (CUSUM) algorithm introduced by \cite{page1955}. The CUSUM algorithm accumulates evidence for a change by summing the log-likelihood ratios of the observations and declares a change when this cumulative sum exceeds a certain threshold.

In modern applications of change detection in science and engineering, two major challenges are encountered: 1) data is high-dimensional, and 2) data is non-stationary. The latter means that the statistical characteristics of the data can change with time. Since it is challenging to learn the distribution of high-dimensional data, it is natural to detect the changes using first or second-order statistics, i.e., to detect the changes in the means and the correlation. 
There exist several approaches for detecting changes in the mean of a distribution (\cite{brodsky1993nonparametric}). In this paper, we address the more challenging problem of detecting a change in correlation in high-dimensional data. When the data distributions are unknown, it is shown in \cite{hou2023robust} and  \cite{unni-etal-ieeeit-2011} that a robust approach often leads to computationally tractable algorithms. Motivated by this, we study robust correlation change detection for high-dimensional data in this paper. 

Directly using estimates of correlation coefficients for change detection in high-dimensional data poses significant challenges (\cite{Fan2014, wainwright2019high}). Specifically, since the data is high-dimensional, we need a data sample at least as big as the dimension of the data to accurately estimate the correlation. This is not a feasible strategy for real-time change detection and will cause a significant delay in detection. This motivates us to utilize a summary statistic that captures the level of correlation 
in the data. We choose to use the maximum magnitude correlation coefficient for this purpose since this summary statistic is well-known to have an asymptotic density in the high-dimensional regime (\cite{bane-hero-tsp-2018, bane-hero-asilomar-2016}). This summary statistics was also used in \cite{bane-hero-tsp-2018} for change detection; however, the solution proposed therein does not apply to time-varying, unknown post-change distributions, and is not robust. 

The paper is organized as follows. After providing some motivating examples in Section~\ref{sec:examples}, we give a brief overview of the classical quickest change detection theory in Section~\ref{sec:ClassicalQCD}. In Section~\ref{sec:NonstationaryQCD}, we discuss the recent result from \cite{hou2023robust} on quickest change detection in non-stationary data. In Section~\ref{sec:SummaryStatAndRobustAlgo}, we discuss our proposed robust algorithm for correlation change detection in high-dimensional data. The asymptotic density of the summary statistic is discussed in Section~\ref{sec:asymDensity} and the change point model and robust algorithm are discussed in 
Section~\ref{sec:RobustAlgo}. In Section~\ref{sec:RobustOptimality}, we prove that the proposed test is indeed robust and optimal. Our proof relies on proving that the densities used in the test are least favorable in a well-defined sense. In Section~\ref{sec:numerical}, we provide applications to simulated datasets to show the robustness and accuracy of the proposed algorithm. Specifically, in Section~\ref{sec:fJaccuracy}, we show that the asymptotic density of the maximum magnitude correlation provides a good approximation to the normalized histogram. In Section~\ref{sec:robustnessOfTest}, we show that the proposed test is robust by applying it to non-stationary data.  In Section~\ref{sec:NonrobustTest}, we compare the performance of the proposed robust test with several non-robust tests. In Section~\ref{sec:nonparametric}, we show that a test based on the asymptotic density is better than a non-parametric test that detects the change in correlation based on a mean change detection algorithm.

\subsection{Some Motivating Applications}
\label{sec:examples}
In this section, we provide some motivating examples of change detection in science and engineering. For more background on change detection theory and its applications, we refer the readers to \cite{veeravalli2014quickest, poor-hadj-qcd-book-2009, tart-niki-bass-2014, tart-book-2019, basseville1993detection, brodsky1993nonparametric}. 
\begin{enumerate}
    \item \textit{Manufacturing Quality Control}:
    In manufacturing processes, ensuring product quality and consistency is significant. Robust correlation change detection can be utilized to monitor the correlation structure of various process parameters or sensor readings in real time. Sudden changes in correlations may indicate abrupt changes in the manufacturing process, allowing for timely intervention to prevent defects or downtime (\cite{Su2019}).

    \item \textit{Environmental Monitoring and Disaster Prevention}: In environmental monitoring applications, such as weather forecasting or natural disaster prediction (\cite{manogaran2018spatial}), detecting sudden changes in correlation patterns among environmental variables can aid in early warning systems and disaster prevention efforts. Correlation change detection can be utilized to monitor correlations among meteorological or environmental parameters, offering insights into potential weather anomalies or environmental hazards.
    
    \item \textit{Financial Market Monitoring}:
     In finance, the detection of sudden changes in market dynamics or trading patterns is crucial for mitigating risks and preventing financial crises. Robust correlation change detection can be employed to promptly identify anomalies in market data, such as unusual correlations between market conditions, asset prices, and volatility (\cite{frisen2009optimal}).

    \item \textit{Healthcare Monitoring and Disease Outbreak Detection}:
    In healthcare, the timely detection of disease outbreaks or abnormal health trends is critical for public health surveillance and response.  Robust correlation change detection can be employed to monitor correlations among various health indicators or disease-related variables, enabling the early identification of emerging health threats and facilitating targeted intervention strategies (\cite{frisen2009optimal}).
\end{enumerate}

\section{Classical Quickest Change Detection Formulation}
\label{sec:ClassicalQCD}
In this section, we briefly review Lorden's classical quickest change detection formulation and its optimal solution, the cumulative sum (CUSUM) algorithm. 
Consider a sequence of independent and identically distributed (i.i.d.) high-dimensional random vectors $\{\mathbf{X}(m)\}$. The vector $\mathbf{X}(m)$ represents our observation at time $m$. Initially, the statistical properties of the process are assumed to be normal and have probability density $f_0$. At a certain unknown point in time $\nu$ called the change point, the density of the vectors changes from $f_0$ to $f_1$: 
\begin{equation}
\label{eq:CPmodel_classical}
    \mathbf{X}(m) \sim \begin{cases}
        f_0, \quad m < \nu, \\
        f_1, \quad m \geq \nu.
    \end{cases}
\end{equation}
The problem of quickest change detection (QCD) is to observe the sequence $\{\mathbf{X}(m)\}$ in real-time and detect this change in density. The change must be detected with the minimum possible detection delay subject to a constraint on the rate of false alarms (\cite{tart-niki-bass-2014, tart-book-2019, poor-hadj-qcd-book-2009, veeravalli2014quickest}). Mathematically, our decision rule is a positive integer-valued random variable $\tau$ that represents the point in time when we declare the change. If after observing $m$ variables we believe that $\nu \leq m$, then the decision variable $\tau$ takes the value $m$. Thus,
$$
1 \leq \tau < \infty. 
$$
In the QCD problem, the variable $\tau$ must be selected to minimize a well-defined metric on the delay $\tau-\nu$. We must also avoid the event of false alarm $\{\nu > \tau\}$. 

If the change point is modeled as a random variable, then the QCD problem can be studied in the Bayesian formulation of \cite{shiryaev1963}. Since a prior distribution of the change point is often not available, we consider a non-Bayesian formulation where the change point $\nu$ is viewed as an unknown constant. The delay metric we consider is the worst-case average detection  delay (WADD) metric proposed by \cite{lord-amstat-1971}:
\begin{equation}\label{eq:lorden}
	\text{WADD}(\tau)  =  \sup_\nu \; \esssup \; \Expect_\nu\left[(\tau - \nu + 1)^+ | \mathbf{X}(1), \mathbf{X}(2), \dots, \mathbf{X}(\nu-1)\right]. 
\end{equation}
Here, $\Expect_\nu$ denotes the expectation conditioned on change occurring at time $\nu$, and $\text{ess sup}$ denotes the essential supremum of a random variable 
defined as follows: for any random variable $Y$,
$$
\text{ess sup}(Y) = \min\{C: \mathsf{P}(Y \leq C) = 1\}. 
$$
Thus, the delay metric WADD represents the worst-case average detection delay, conditioned on the worst possible realizations of observations before change. Since the change point $\nu$ is an unknown constant, assuming any fixed choice of the change point $\nu$ is unreasonable. Thus, in the WADD metric, the worst-case value of the change point is also chosen. The variable $\tau$ must be selected to minimize WADD$(\tau)$.  However, an additional penalty or constraint must be imposed on the rate of false alarms to avoid a trivial solution (for example, we can set $\tau=0$ to minimize WADD). The metric 
chosen for the rate of false alarms in \cite{lord-amstat-1971} is its inverse, the mean time to a false alarm (MFA):
\begin{equation}
\label{eq:MFA}
 \text{MFA}(\tau) = \mathsf{E}_\infty[\tau].
\end{equation}
Here $\mathsf{E}_\infty$ is the expectation conditioned on the fact that there is no change point (i.e., the change occurs at time infinity). The formulation studied by \cite{lord-amstat-1971} is the following:
\begin{equation}\label{eq:optim}
\begin{split}
    \min_\tau &\; \quad \text{WADD}(\tau), \\
    \text{subj. to}&\; \quad \text{MFA}(\tau) \geq \beta,
    \end{split}
\end{equation}
where $\beta > 0$ is a given desired lower bound on the MFA. Since the goal is to minimize the maximum of delay, this is often called a minimax formulation in the literature. The decision time $\tau$
must be chosen so that the change detection can be performed in real time. Thus, it is assumed that $\tau$ is a stopping time, that is, the event $\{\tau \leq m\}$ depends only on the observations up to time $m$. 

The optimal solution for this problem formulation is the popular cumulative sum (CUSUM) algorithm (\cite{page-biometrica-1954}). In this algorithm, we compute a sequence of statistics $\{W_m\}_{m \geq 0}$ as follows: start with $W_0=0$, and set
\begin{equation}
\label{eq:CUSUM}
    W_m = \max\left\{0, W_{m-1} + \log \frac{f_1(\mathbf{X}(m))}{f_0(\mathbf{X}(m))}\right\}.
\end{equation}
A change is declared the first time the statistic $W_m$ crosses a pre-defined threshold $A$:
\begin{equation}
    \tau_c = \min\{m \geq 1: W_m > A\}.
\end{equation}
The value of $A$ is chosen to control the rate of false alarms, i.e., to satisfy the constraint on the MFA in \eqref{eq:optim}. While this algorithm was proposed in \cite{page-biometrica-1954}, its optimality was established in \cite{lord-amstat-1971} and \cite{moustakides1986} many years later. It is also shown in \cite{lord-amstat-1971} that
setting $A=\log \beta$ is sufficient to guarantee that 
\begin{equation}
    \text{MFA}(\tau_c) \geq \beta. 
\end{equation}
Furthermore, for this choice of threshold,
\begin{equation}
    \text{WADD}(\tau_c) = \frac{\log \beta}{D(f_1 \; \| \; f_0)} (1 + o(1)), \quad \beta \to \infty. 
\end{equation}
Here $o(1) \to 0$ in the limit and $D(f_1 \; \| \; f_0)$ is the Kullback-Leibler (KL) divergence between the density $f_1$ and $f_0$ defined as 
\begin{equation}
    D(f_1 \; \| \; f_0) = \int_x f_1(x) \log \frac{f_1(x)}{f_0(x)} dx. 
\end{equation}
The KL divergence has the following properties: $D(f_1 \; \| \; f_0)  \geq 0$ and $D(f_1 \; \| \; f_0) = 0$ if, and only if, $f_1=f_0$. It is not symmetric and hence is not a true metric between densities. However, if the KL divergence is seen as a pseudo distance between the densities, then the above equation shows that the delay of the CUSUM algorithm depends inversely on this distance.

\section{Robust Quickest Change Detection in Non-Stationary Processes}
\label{sec:NonstationaryQCD}
The classical change point model given in \eqref{eq:CPmodel_classical} has two major limitations. First, in many engineering applications, while the pre-change density can be assumed to be fixed, the post-change density is often non-stationary, i.e., changes with time. Second, while the pre-change density can be assumed to be known, the post-change density is often unknown. Thus, a more realistic change point model is 
\begin{equation}\label{eq:CPmodel_robust}
\mathbf{X}(m) \sim
	\begin{cases}
		f_0, &\quad \forall m < \nu, \\
		f_{m, \nu}, &\quad \forall m \geq \nu.
	\end{cases}
\end{equation}
Thus, before the change point $\nu$, the random variables $\{\mathbf{X}(m)\}$ are i.i.d. with density $f_0$. After $\nu$, the variables are independent with $\mathbf{X}(m)$ with density $f_{m, \nu}$. Note that in this model, the density at time $m$ depends on the location of the change point $\nu$. The decision-maker knows $f_0$, but may not know $\{f_{m, \nu}\}$. 

We now discuss a robust solution to this problem studied in \cite{hou2023robust}. In \cite{hou2023robust}, the unknown density sequence $\{f_{m, \nu}\}$ is assumed to belong to a \textit{known} family of densities $\{\mathcal{P}_{m, \nu}\}$ such that
$$
f_{m, \nu} \in \mathcal{P}_{m, \nu}, \quad m,\nu = 1,2, \dots.
$$
We denote the set of all possible post-change density sequences by 
$$
\mathcal{G} = \{G = \{f_{m, \nu}\}: f_{m, \nu} \in \mathcal{P}_{m, \nu}, \; m, \nu \geq 1\}.
$$
With this notation, the robust problem formulation studied in \cite{hou2023robust} is
\begin{equation}\label{eq:RobustProblem}
\begin{split}
    \min_\tau &\; \quad \sup_{G \in \mathcal{G}} \; \text{WADD}^{G}(\tau), \\
    \text{subj. to}&\; \quad \text{MFA}(\tau) \geq \beta,
    \end{split}
\end{equation}
where $\text{WADD}^{G}(\tau)$ is the WADD delay measure when the post-change density sequence is $G=\{f_{n, \nu}\}$. It is proved in \cite{hou2023robust} that the optimal solution to the above problem is a robust version of the CUSUM algorithm in \eqref{eq:CUSUM}. 

To state this more precisely, we need some definitions. We say that a random variable $Z_2$ is stochastically larger than another random variable $Z_1$ if
$$
\Prob(Z_2 \geq t) \geq \Prob(Z_1 \geq t), \quad \forall t \in (-\infty, \infty). 
$$
We use the notation
$$
Z_2 \succ Z_1.
$$
If $\mathcal{L}_{Z_2}$ and $\mathcal{L}_{Z_1}$ are the probability laws of $Z_2$ and $Z_1$, then we also use the notation
$$
\mathcal{L}_{Z_2} \succ \mathcal{L}_{Z_1}
$$
to state that $Z_2 \succ Z_1$. 
We also use
$$
\mathcal{L}(\phi(\mathbf{X}), g)
$$
to denote the law of some function $\phi(\mathbf{X})$ of the random variable $\mathbf{X}$, when the variable $\mathbf{X}$ has density $g$.
We say that the family $\{\mathcal{P}_{m, \nu}\}$ is stochastically bounded by the sequence
	$$
	\bar{G}=\{\bar{f}_{m, \nu}\},
	$$
	and call $\bar{G}$ the least favorable law (LFL), if
	$$
	\bar{f}_{m, \nu} \in \mathcal{P}_{m, \nu}, \quad \nu, m=1,2, \dots, 
	$$
	and $\forall  f_{m, \nu} \in \mathcal{P}_{m, \nu},  m, \nu=1,2, \dots$,
	\begin{equation}
 \label{eq:stocbounded}
		\begin{split}
			\mathcal{L}\left(\log \frac{\bar{f}_{m, \nu }(\mathbf{X}(m))}{f_0(\mathbf{X}(m))}, f_{m, \nu}\right) &\succ 	\mathcal{L}\left(\log \frac{\bar{f}_{m, \nu}(\mathbf{X}(m))}{f_0(\mathbf{X}(m))},
			\bar{f}_{m, \nu}\right).
		\end{split}
	\end{equation}
If the LFL exists, then it has been proved in \cite{hou2023robust} that the CUSUM algorithm designed using the LFL sequence $\{\bar{f}_{m, \nu}\}$ is optimal. For simplicity of notation, we consider a special case (relevant to this paper) where
$$
\bar{f}_{m, \nu} = \bar{f}_1, \quad \forall m, \nu. 
$$
In this case, the robust solution simplifies to the following: the robust CUSUM statistic is 
\begin{equation}
\label{eq:CUSUM_robust}
    \bar{W}_m = \max\left\{0, \bar{W}_{m-1} + \log \frac{\bar{f}_1(\mathbf{X}(m))}{f_0(\mathbf{X}(m))}\right\}, \quad \bar{W}_0 = 0, 
\end{equation}
and the change is declared the first time the statistic $\bar{W}_m$ crosses a pre-defined threshold:
\begin{equation}
    \tau_{cr} = \min\{m \geq 1: \bar{W}_m > A\}.
\end{equation}
The threshold $A$ must be chosen to satisfy the constraint on the MFA. The threshold for the robust CUSUM algorithm can be designed similarly: setting $A=\log \beta$ is sufficient to guarantee that 
\begin{equation}
    \text{MFA}(\tau_{cr}) \geq \beta. 
\end{equation}

\section{Algorithm for Robust Quickest Correlation Change Detection}
\label{sec:SummaryStatAndRobustAlgo}
When the dimension of the variables $\{\mathbf{X}(m)\}$ is large, learning the pre- and post-change densities is not feasible. Thus, the CUSUM algorithm \eqref{eq:CUSUM} cannot be directly applied to the sequence $\{\mathbf{X}(m)\}$. When the dimension of the variables is large, it is also hard to identify uncertainty classes for the densities and identify a least favorable density sequence. Thus, even the robust CUSUM algorithm using the statistic in \eqref{eq:CUSUM_robust} cannot be directly applied to the sequence $\{\mathbf{X}(m)\}$. 

In the absence of information on the distribution, the next logical approach in the literature is often to detect changes using first-order (mean) and second-order (correlation) statistics. In this paper, we focus on the problem of correlation change detection. Since the detection must be done in real time and the data is high-dimensional, we cannot afford to use correlation or covariance estimation techniques followed by detection. Specifically, if the data is $p$-dimensional, we will need at least $p$ samples to reliably estimate the correlation, which adds a delay of $p$ samples to the test. 

This motivates us to follow the approach taken in \cite{bane-hero-tsp-2018}; here, a change in correlation is detected using the maximum magnitude correlation coefficient. This quantity also captures the level of correlation in the data. Under distributional and sparsity assumptions, an asymptotic density is obtained for the quantity. This density is then used for change detection. Our principal contribution in this paper is that we obtain a \textbf{robust} test for maximum correlation change detection using the above-mentioned asymptotic density. In the remainder of this section, we describe the asymptotic density and the robust test in detail. The asymptotic density is discussed in Section~\ref{sec:asymDensity}. The change point model for correlation detection and the proposed algorithm are discussed in Section~\ref{sec:RobustAlgo}. The robust optimality of the proposed algorithm for the problem formulation specified in \eqref{eq:RobustProblem} is proved in Section~\ref{sec:RobustOptimality}.

\subsection{Asymptotic Density for Maximum Magnitude Correlation}
\label{sec:asymDensity}


Let $\mathbb{X}$ be a $n \times p$ matrix of $n$ vectors $\mathbf{X}(1), \dots, \mathbf{X}(n)$, each of dimension $p$:  
$$\mathbb{X} = [\mathbf{X}_1, \cdots, \mathbf{X}_p] = [\mathbf{X}(1)^T, \cdots, \mathbf{X}(n)^T]^T,$$
where $\mathbf{X}_i = [X_{1i}, \cdots, X_{ni}]^T$ is the $i^{th}$ column containing the $i^{th}$ component of all the $n$ vectors and $\mathbf{X}(i) = [X_{i1}, \cdots, X_{ip}]$
is the $i^{th}$ row.
It is assumed that the vectors $\mathbf{X}(1), \dots, \mathbf{X}(n)$ are independent and have elliptically contoured distribution or
elliptical density (\cite{ander-mutlstat-1996}):
\begin{equation}\label{eq:matrixdensityexpr}
f_{\mathbf{X}}(\mathbf{x}) = h((\mathbf{X}-{\mathbf \mu} )^T\boldsymbol{\Sigma}^{-1}(\mathbf{X}-{\mathbf \mu} )^T),
\end{equation}
for shaping function $h$ on $\mathbb{R}^+$. The elliptical density family includes multivariate Gaussian as a special case with $h(x)=(2\pi)^{-k/2} \text{det}(\boldsymbol{\Sigma})^{-1/2}e^{-x/2}$. 

Define the $p \times p$ sample covariance matrix as
$$ \mathbf{S} = \frac{1}{n-1} \sum_{i=1}^n (\mathbf{X}(i) - \bar{\mathbf{X}})^T (\mathbf{X}(i) - \bar{\mathbf{X}}),$$
where $\bar{\mathbf{X}}$ is the sample mean of the $n$ rows of $\mathbb{X}$.
Define the sample correlation (coherency) matrix as
$$ \mathbf{R} = \mathbf{D_S}^{-1/2} \mathbf{S} \mathbf{D_S}^{-1/2},  $$
where $\mathbf{D_A}$ denotes the matrix obtained by zeroing out all but the diagonal elements of the matrix $\mathbf{A}$.
For any $i$ and $j$, $\mathbf{R}_{ij}$, the $ij$th element of the matrix $\mathbf{R}$, is the
sample correlation coefficient between the $i^{th}$ and $j^{th}$ columns of $\mathbb X$.

Define the maximum magnitude correlation coefficient as
\begin{equation}\label{eq:Sum_Stat_V}
V(\mathbb{X}):= \max_{1 \leq i \leq p} \max_{j \neq i} |\mathbf{R}_{ij}|.
\end{equation}
The following theorem is proved in \cite{bane-hero-tsp-2018}. For the theorem, we note that a matrix is called row sparse of degree $s$ if it has at most $s$ non-zero terms in each row with $s=o(p)$. A matrix is called block sparse 
of degree $s$, if there exists a row-column
permutation that puts it into block diagonal form with block size of size $s \times s$, where $s=o(p)$.

\begin{theorem}[\cite{bane-hero-tsp-2018}]\label{thm:CorrScr}
Let $\mathbf{\Sigma}$, the population dispersion matrix of the rows of $\mathbb{X}$, be row-sparse of degree $s=o(p)$. Let $p \to \infty$ and $\rho = \rho_p \to 1$ such that
$p (p-1)(1-\rho^2)^{(n-2)/2} \to e_{n} \in (0, \infty)$, where $e_n$ is a constant.  
Then:
\begin{enumerate}
\item As $p\rightarrow \infty$:
$$\Prob(V(\mathbb{X}) \geq \rho) \to 1-\exp(-\Lambda J),$$
where
$J$ is a fixed scalar depending on the distribution of the sample correlation matrix $\mathbf R$ 
and
$$\Lambda = \lim_{p\to \infty, \rho \to 1} \Lambda(\rho) = ((e_{n} a_n ) / (n-2) ),$$
with
$$\Lambda(\rho) = p {p-1 \choose 1} P_0(\rho),$$
$$P_0(\rho) = a_n \int_{\rho}^1 (1-u^2)^{\frac{n-4}{2}} du,$$
$$a_n = 2B((n-2)/2, 1/2) \mbox{ with } B(l,m) \mbox{ the beta function}. $$

\item If $\mathbf{\Sigma}$ is block sparse of degree $s$,
then
$$J=1 + O((s/p)^{k+1}),$$
and if $\mathbf{\Sigma}$ is diagonal then $J=1$.
\end{enumerate}
\end{theorem}

Using Theorem~\ref{thm:CorrScr}, for large values of dimension $p$, the cumulative distribution function of $V(\mathbb{X})$ can be approximated by
\begin{equation}\label{eq:CDF_V}
\Prob(V(\mathbb{X}) \leq \rho) = \exp(-\Lambda(\rho) \; J), \; \quad \rho \in [0,1],
\end{equation}
where $\Lambda(\rho)$ is as defined in Theorem~\ref{thm:CorrScr}.
While this approximation is valid for large values of $\rho$, numerical experiments suggest (see Section~\ref{sec:numerical}) that the approximation is valid uniformly over $\rho$, as long as $n$ is small and $p\gg n$.

Taking derivatives on both sides we get an approximate density for the variable 
\( V(\mathbb{X}) \) in a parametric family parametrized by the parameter $J$:
\begin{equation}
\label{eq:PDF_V_delta1}
f_V(v; J) = \frac{C}{2} (1-v^2)^{\frac{n-4}{2}} J \exp\left(-\frac{C}{2} J \; T(v) \right), \; v \in (0,1].
\end{equation}
Here, if $B(l,m)$ is the beta function then,
\begin{equation}\label{eq:C_pndelta}
C = 2 p (p-1) B((n-2)/2, 1/2),
\end{equation}
and
\begin{equation}\label{eq:Trho}
T(v) = \int_{v}^1 (1-u^2)^{\frac{n-4}{2}} du
\end{equation}
is the incomplete beta function. We recall from the theorem above that if the covariance matrix $\mathbf{\Sigma}$ is diagonal then $J=1$. In the rest of the paper, while stating mathematical statements, we treat this approximate density as the true density of \( V(\mathbb{X}) \). We also remark that the density expression is valid for any $n \geq 5$. Finally, we note that in \cite{bane-hero-tsp-2018} and \cite{bane-hero-asilomar-2016}, the authors have also developed an asymptotic theory for correlation that goes well beyond the maximum magnitude correlation coefficient. Specifically, the $k^{th}$ order maximum magnitude correlation is also studied. The detection theory developed below can also be applied to these more general maximal coefficients. 

\subsection{Correlation Change Point Model and Robust Algorithm}
\label{sec:RobustAlgo}
Given a stochastic sequence of vectors $\{\mathbf{X}(\ell)\}$, our goal is to detect a change in the distribution of these vectors. Since the vectors are high-dimensional, we cannot learn the distribution. As mentioned above, our goal is to detect changes in the second-order statistics (correlation). The high dimensionality of the data again motivates us to seek summary statistics for the correlation rather than aim for detection after correlation estimation. The development in the previous section motivates us to choose the maximum magnitude correlation coefficient as our summary statistic. 

To detect the change using the maximum magnitude correlation coefficients, we collect the vectors $\{\mathbf{X}(\ell)\}$ in the form of non-overlapping $n \times p$ matrices $\{\mathbb{X}(m)\}$ and obtain the maximum magnitude correlation coefficient sequence $\{V(\mathbb{X}(m))\}$. In the rest of the paper, we use the notation
$$
V_m := V(\mathbb{X}(m)).  
$$

The change detection problem we solve in terms of the sequence $\{V_m\}$ is as follows:
\begin{equation}\label{eq:corrChgmodel}
	V_m \sim
	\begin{cases}
		f_V(v, 1), &\quad \forall m < \nu, \\
		f_V(v, J_m), &\quad \forall m \geq \nu.
	\end{cases}
\end{equation}
Here $f_V$ is the density given in \eqref{eq:PDF_V_delta1}. Thus, the problem we consider is that before the change, the variables are uncorrelated with the parameter $J=1$. After the change time $\nu$, the variables become correlated with an unknown and time-varying or non-stationary correlation level. This is captured in the model by selecting a time-varying $J$ parameter sequence $\{J_m\}$. This brings us to the QCD problem discussed in Section~\ref{sec:NonstationaryQCD}. 

Let us assume that there is a known $\bar{J} > 1$ such that 
\begin{equation}
    J_m \geq \bar{J} > 1. 
\end{equation}
We will show in Section~\ref{sec:RobustOptimality} that the following test is exactly robust optimal for the robust problem in \eqref{eq:RobustProblem}: start with $W_0=0$ and update $W_m$ using
\begin{equation}\label{eq:VRobustAlgo}
W_m = \max \left\{0, W_{m-1} + \log \frac{f_V(V_m, \bar{J})}{f_V(V_m, 1)}\right\}. 
\end{equation}
Again, a change is declared and an alarm is raised at
\begin{equation}
    \tau_{cc} = \left\{m\geq 1: W_m \geq A \right\}.
\end{equation}
Recall that setting $A=\log \beta$ will guarantee that $\text{MFA}(\tau_{cc}) \geq \beta$. Finally, when there is no clear choice for $\bar{J}$, one can select a putative value of $\bar{J} = 2$.

\section{Robust Optimality of the Proposed Correlation Change Detection Test}
\label{sec:RobustOptimality}
In this section, we prove that the proposed algorithm in the previous section in \eqref{eq:VRobustAlgo} is robust optimal for the problem formulation in \eqref{eq:RobustProblem}. We note that this optimality is valid under two assumptions: 1) The observation sequence is now $\{V_m\}$ and not $\{\mathbf{X}(\ell)\}$, and 2) the density $f_V(v, J)$ is a true density for the maximum magnitude correlation coefficient variable $V$. It is important to note that the optimality result proved in the following theorem is non-asymptotic, i.e., it is exact for any value of $\beta$ in \eqref{eq:RobustProblem}. Thus, it constitutes one of the first results on exact robust optimality for correlation change detection in a non-stationary setting. 

We need two supporting lemmas for our proof. 

\begin{lemma}\label{lem:unni}
    Suppose $\{U_i: 1 \leq i \leq n\}$ is a set of mutually independent random variables, and $\{V_i: 1 \leq i \leq n\}$ is another set of mutually independent random variables such that $U_i \succ V_i$, $1 \leq i \leq n$. Now let $q: \mathbb{R}^n \to \mathbb{R}$ be a continuous real-valued function defined on $\mathbb{R}^n$ that satisfies
    \begin{align*}
        q(X_1, \dots, X_{i-1}, a, X_{i+1}, \dots, X_{n}) \geq q(X_1, \dots, X_{i-1}, X_i, X_{i+1}, \dots, X_{n})
    \end{align*}
    for all $(X_1, \dots, X_n) \in \mathbb{R}^n$, $a > X_i$, and $i \in \{1, \dots, n\}$. Then we have 
    \begin{align*}
        q(U_1, U_2, \dots, U_n) \succ q(V_1, V_2, \dots, V_n).
    \end{align*}
\end{lemma}

\begin{lemma}
\label{lem:MLRorder}
    Let $f$ and $g$ be two probability density functions such that $f$ dominates $g$ in monotone likelihood ratio (MLR) order:
    $$
    \frac{f(x)}{g(x)} \; \; \uparrow \; \; x, 
    $$
    i.e., the likelihood ratio $\frac{f(x)}{g(x)}$ is monotonically increasing in the variable $x$. 
    Then, $f$ also dominates $g$ in stochastic order, i.e., $f \succ g$ or 
    $$
    \int_{x}^\infty f(y) dy \geq \int_x^\infty g(y) dy, \quad \forall x. 
    $$
\end{lemma}

The proof of Lemma~\ref{lem:unni} can be found in \cite{unni-etal-ieeeit-2011} and that of Lemma~\ref{lem:MLRorder} can be found in \cite{hou2023robust} or \cite{krishnamurthy2016partially}. We now state and prove the main theorem of this paper.

\begin{theorem}
    If $J_m \geq \bar{J}$ for all $m$, then the density $f_V(v, \bar{J})$ is the LFL for the family $\{f_V(v, J_m)\}$, and therefore, the CUSUM algorithm in \eqref{eq:VRobustAlgo} is exactly robust optimal for the problem formulation in \eqref{eq:RobustProblem} when the change-point model is \eqref{eq:corrChgmodel}. 
\end{theorem}
\begin{proof}
To prove this theorem, we need to verify the conditions in \eqref{eq:stocbounded}. 
    We note that
    \begin{equation}
        \frac{f_V(v, \bar{J})}{f_V(v, 1)} = \frac{\frac{C}{2} (1-v^2)^{\frac{n-4}{2}} \bar{J} \exp\left(-\frac{C}{2} \bar{J} \;T(v) \right)}{\frac{C}{2} (1-v^2)^{\frac{n-4}{2}}  \exp\left(-\frac{C}{2}  \;T(v) \right)}
        = \bar{J} \exp\left(-\frac{C}{2} (\bar{J}-1) \;T(v) \right).
    \end{equation}
    Taking the logarithm, we get
     \begin{equation}
       \log  \frac{f_V(v, \bar{J})}{f_V(v, 1)} = \log \bar{J} -\frac{C}{2} (\bar{J}-1) \;T(v).
    \end{equation}
We need to show that 
    \begin{equation}
        \mathcal{L}\left(\log \bar{J} -\frac{C}{2} (\bar{J}-1) \;T(V), \; f_V(v, J_m) \right) \succ  \mathcal{L}\left(\log \bar{J} -\frac{C}{2} (\bar{J}-1) \;T(V), \; f_V(v, \bar{J}) \right)
    \end{equation}
    To this end, note that $T(v) = \int_{v}^1 (1-u^2)^{\frac{n-4}{2}} du$ is decreasing in $v$. This means $\log \bar{J} -\frac{C}{2} (\bar{J}-1) \;T(v)$ is increasing in $v$. This function is also continuous. Hence due to the Lemma~\ref{lem:unni}, it is enough to show that  
    \begin{equation}
        \label{eq:JmsuccJbar}
        f_V(v, J_m) \succ f_V(v, \bar{J}).
    \end{equation}
    But, note that for any $J_m \geq \bar{J}$, the likelihood ratio
    \begin{equation}
        \frac{f_V(v, J_m)}{f_V(v, \bar{J})} = \frac{\frac{C}{2} (1-v^2)^{\frac{n-4}{2}} J_m \exp\left(-\frac{C}{2} J_m \;T(v) \right)}{\frac{C}{2} (1-v^2)^{\frac{n-4}{2}}  \bar{J}\exp\left(-\frac{C}{2} \bar{J} \;T(v) \right)}
        = \frac{J_m}{\bar{J}} \exp\left(-\frac{C}{2} (J_m-\bar{J}) \;T(v) \right)
    \end{equation}
    is monotonically increasing in $v$. This proves \eqref{eq:JmsuccJbar} and hence the theorem due to Lemma~\ref{lem:MLRorder}. 
    
\end{proof}

\section{Numerical Results}
\label{sec:numerical}

In this section, we show the effectiveness and robustness of our proposed test \eqref{eq:VRobustAlgo}. We use a multivariate Gaussian distribution to generate data from elliptically contoured density. By the central limit theorem, we expect that the proposed test will work well for a broad class of distributions, even if they are not in the elliptically contoured family.  



We generate the data following the methodology outlined in \cite{bane-hero-tsp-2018}, using $n=10$ and $p=100$. This choice of $n$ is arbitrary. However, note that the value of $n$ controls how the online stream of high-dimensional vectors is grouped into batches. While any $n \geq 5$ can be chosen in principle, a large value of $n$ will cause a large detection delay. 

We first generate covariance matrices by sampling from the Wishart distribution. We then make them row-sparse by amount $s=5$ using a pre-designed mask (the definition of row-sparsity is given before Theorem~\ref{thm:CorrScr}). We then generate samples from a multivariate Gaussian vector with the row-sparse covariance matrix. 



\subsection{Accuracy of Asymptotic Density}
\label{sec:fJaccuracy}
In this section, we show that the asymptotic density \eqref{eq:PDF_V_delta1} provides a good approximation to the true density of the maximum magnitude correlation coefficient. We generate multivariate Gaussian vectors using the methodology discussed above, group them into batches of size $n=10$, and calculate the value of the variable $V(\mathbb{X})$. For each covariance matrix, we generate $T=5000$ samples of $V$ and plot the normalized histogram to get an estimated density for $V$. 

We next estimate the parameter $J$ by using maximum likelihood estimation, i.e.,  by maximizing $\prod_{m=1}^{T} f_V(V_m;J)$, using (\ref{eq:PDF_V_delta1}). 
The log-likelihood function is given by
\begin{equation}\label{eq:log_likelihood}
\ell(J) = \sum_{m=1}^{T} \log f_V(V_m;J).
\end{equation}
Substituting for the density by the expression in (\ref{eq:PDF_V_delta1}) we have:
\begin{equation}\label{eq:density_substitution_MLE}
\ell(J) = \sum_{m=1}^{T} \log \left(\frac{C}{2} (1-(V_m^2)^{\frac{n-4}{2}} J \exp\left(-\frac{C}{2} J \;T(V_m) \right) \right).
\end{equation}
This simplifies to:
\begin{equation}\label{eq:density_substitution_MLE_simplified}
\ell(J) = T \log \left(\frac {C}{2}\right) + \sum_{m=1}^{T} \left(\frac {n-4}{2} \log (1-(V_m)^2) + \log (J) - \frac{C}{2} \; J \; T(V_m) \right).
\end{equation}
Differentiating $\ell(J)$ with respect to $J$ yields:
\begin{equation}\label{Differentiation}
\frac{d\ell(J)}{dJ} = \sum_{m=1}^{T} \left(\frac{1}{J} - \frac{2}{C} \; T(V_m) \right) = 0.
\end{equation}
Solving for $J$ we obtain:
\begin{equation}\label{solving_J}
J = \frac{1}{\frac{C}{2} \frac{1}{T} \sum_{m=1}^{T} T(V_m)}.
\end{equation}
Thus, the estimated value of the parameter $J$, given the $T$ samples $(V_1, \ldots, V_T)$ from the distribution $f_V (\cdot,J)$ is given by:
\begin{equation}\label{J_estimator}
\hat{J}(V_1, \ldots, V_T) = \frac{1}{\frac{C}{2} \frac{1}{T} \sum_{m=1}^{T} T(V(m))}. 
\end{equation}

The Wishart covariance matrix is resampled multiple times to get different values of the parameter $J$: $J = 1.0, \; 2.79, \; 3.62, \; 7.23, \text{ and } 11.12$. 
The normalized histograms and the plot for the density using the estimated $J$ values are then plotted together. The results are reported in 
Figure \ref{fig:histogram_vs_density}. 
The figures show that the density (\ref{eq:PDF_V_delta1}) indeed provides a good approximation to the normalized histogram. 

\begin{figure}
\centering \includegraphics[scale=0.8]{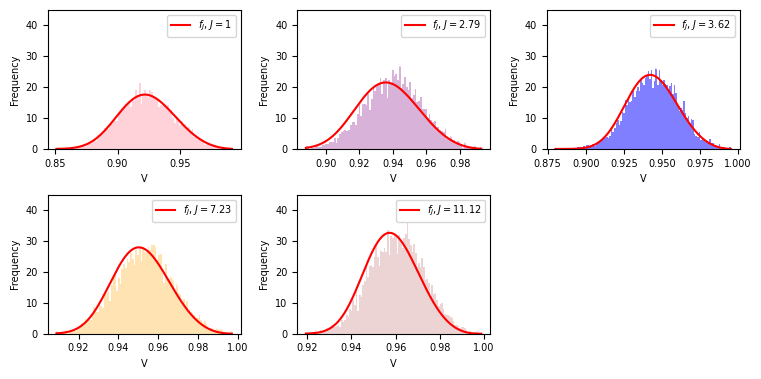}
\caption{Comparison between normalized histograms and asymptotic density in \eqref{eq:PDF_V_delta1} for $J$ = $1.0, \; 2.79, \; 3.62, \; 7.23, \text{and } 11.12$ respectively. } \label{fig:histogram_vs_density}
\end{figure}

\subsection{Robustness of the Proposed Test}
\label{sec:robustnessOfTest}
We now demonstrate the robustness of the proposed test \eqref{eq:VRobustAlgo} with $\bar{J}=2$ for non-stationary post-change data as illustrated in Figure ~\ref{fig:Non-stationary}. The data for the test is generated as follows. The first $500$ $V$ values are 
generated with $J_0=1$. The post-change $V$ values from 501 to 2000 are generated by varying $J$ values across different intervals. Specifically, after the change point at time step 500, the $V$ values in these intervals have estimated $J$ values as follows:  
\begin{itemize}
    \item from 500 to 800, shown in purple, $J=7.23$. 
    \item from 800 to 1300, shown in orange, $J=11.12$
    \item from 1300 to 1670, shown in blue,  $J=3.62$
    \item from 1670 to 2000, shown in teal,  $J=2.79$.
\end{itemize}

\begin{figure}
\centering
\includegraphics[width=0.75\linewidth, height=0.43\linewidth]{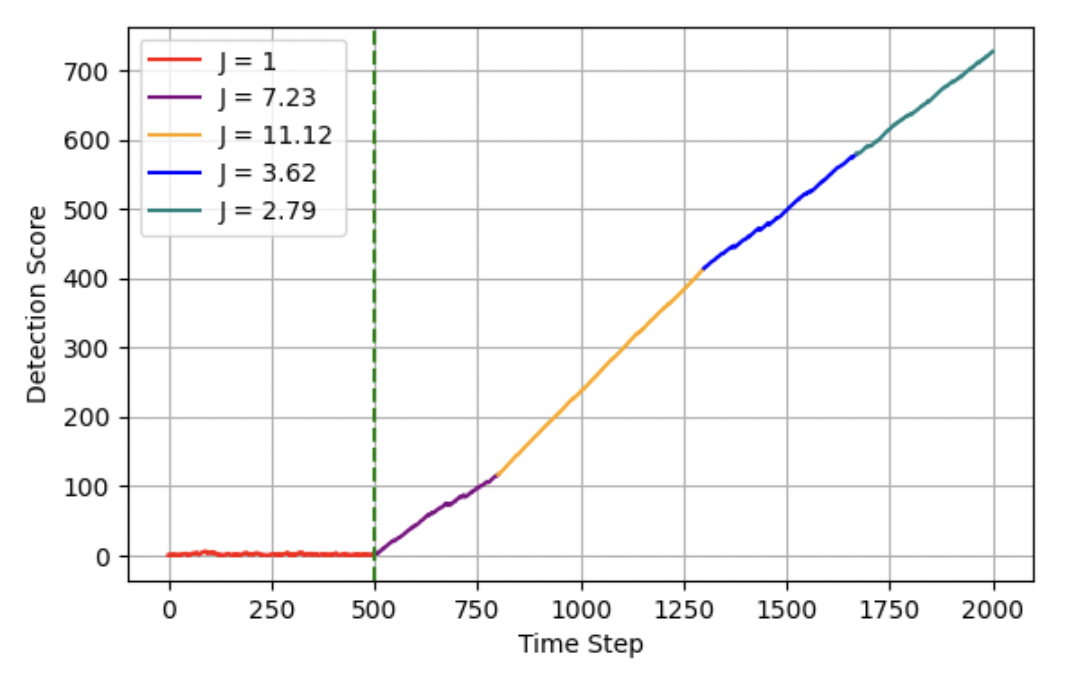}
\caption{\label{fig:Non-stationary}Robust test applied to non-stationary data with a change point at $500$. }
\end{figure}
Despite the non-stationarity in the post-change environment, the robust CUSUM test \eqref{eq:VRobustAlgo}, with $\bar{J}=2$, effectively detected the change. The evolution of the test statistic is plotted in Figure ~\ref{fig:Non-stationary}. Before the change, the statistic is close to zero. After the change, the statistic grows towards infinity, which can be detected using a threshold. Treating such a change in the drift as a measure of successful change detection is a standard approach in the quickest change detection literature (\cite{tart-niki-bass-2014, tart-book-2019, poor-hadj-qcd-book-2009, veeravalli2014quickest}).

\subsection{Comparison with Non-Robust Tests}
\label{sec:NonrobustTest}
In this section, we continue our discussion of the robustness of our proposed test \eqref{eq:VRobustAlgo} by comparing it with other algorithms that are not necessarily designed to be robust. Recall that our proposed robust test statistic is given by 
\begin{equation}\label{eq:VRobustAlgo_1}
W_m = \max \left\{0, W_{m-1} + \log \frac{f_V(V_m, \bar{J})}{f_V(V_m, 1)}\right\}. 
\end{equation}
Also, recall that the post-change $J$ parameters $\{J_m\}$ are all assumed to be greater than $\bar{J}$. We set $\bar{J}=2$. In what follows, a non-robust test uses the statistic
\begin{equation}\label{eq:VRobustAlgo_2}
W_m = \max \left\{0, W_{m-1} + \log \frac{f_V(V_m, J_1)}{f_V(V_m, 1)}\right\},
\end{equation}
where the parameter $J_1 \neq \bar{J}$ and in fact can be much bigger.  Both tests stop and declare a change when their corresponding $W_m$ statistic crosses a pre-defined threshold. 

In Figures ~\ref{fig:J2.79}, \ref{fig:J3.62}, \ref{fig:J7.23}, and \ref{fig:J11.12}, we compare the performance of robust and non-robust CUSUM algorithms. 
In each of the four figures, it can be seen that the robust CUSUM algorithm shows a clear increase in the detection score starting at \( t = 500 \), indicating effective change detection, while the non-robust CUSUM fails to detect the change consistently.

For Figures~\ref{fig:J2.79} and \ref{fig:J3.62}, the non-robust test is desinged using \( J_1 = 10 \). 
In Figure~\ref{fig:J2.79}, the post-change $V$ values have estimated $J=2.79$ while in Figure~\ref{fig:J3.62}, the post-change $V$ values have estimated $J=3.62$. As seen in these two figures, the statistics stay close to zero, even after the change, signifying a failure to detect the change. 
However, the statistic for Figure~\ref{fig:J3.62} shows fluctuations due to an increase in the variance. 

\begin{figure}
\centering
\includegraphics[width=0.95\linewidth, height=0.35\linewidth]{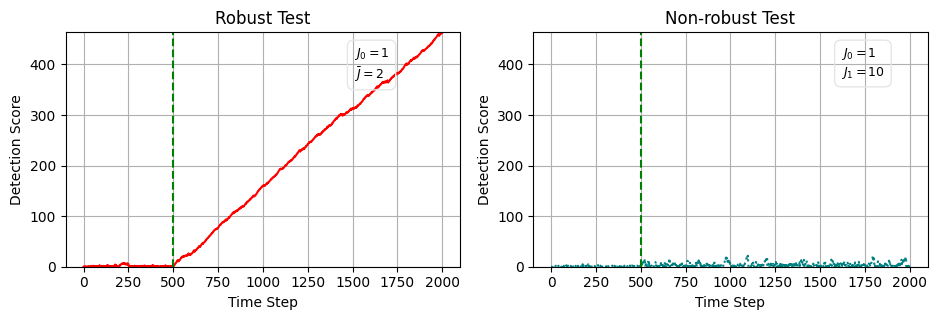}
\caption{Comparison between a robust and a non-robust CUSUM algorithm. The post-change data samples have an estimated  $J = 2.79$.}
\label{fig:J2.79}
\end{figure}

\begin{figure}
\centering
\includegraphics[width=0.95\linewidth, height=0.35\linewidth]{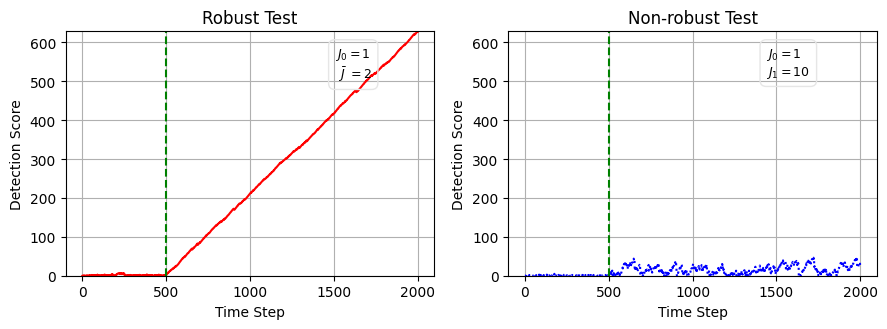}
\caption{Comparison between a robust and a non-robust CUSUM algorithm. The post-change data samples have an estimated $J = 3.62$.}
\label{fig:J3.62}
\end{figure}
In Figure~\ref{fig:J7.23}, the non-robust CUSUM algorithm is designed using parameter \( J_1 > 15 \) and the post-change data is generated using $J=7.23$. For Figure  \ref{fig:J11.12}, \( J_1 > 20 \), and the data is generated using $J=11.12$. Again, the statistics for the non-robust tests stay close to zero even after the change, but displays even greater variability with peaks around 150 and 200, respectively





\begin{figure}
\centering
\includegraphics[width=0.95\linewidth, height=0.35\linewidth]{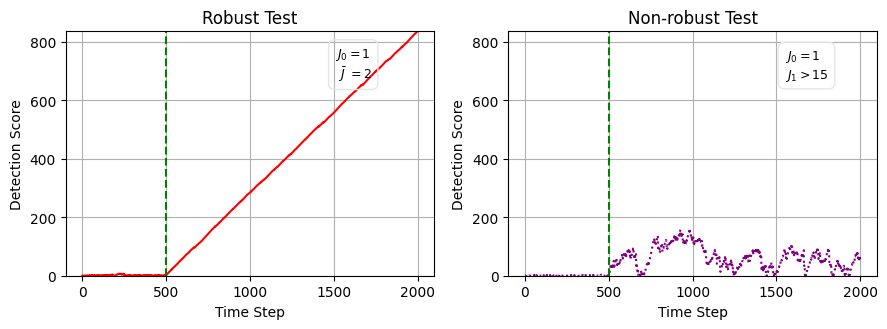}
\caption{Comparison between a robust and a non-robust CUSUM algorithm. The post-change data samples have an estimated $J = 7.23$.}
\label{fig:J7.23}
\end{figure}

\begin{figure}
\centering
\includegraphics[width=0.95\linewidth, height=0.35\linewidth]{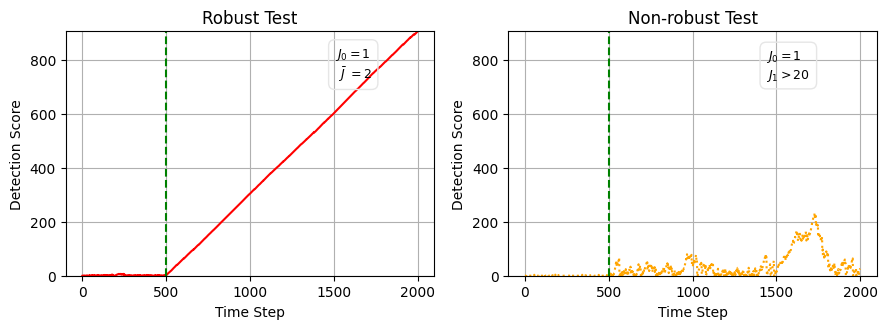}
\caption{Comparison between a robust and a non-robust CUSUM algorithm. The post-change data samples have an estimated $J = 11.12$.}
\label{fig:J11.12}
\end{figure}

Finally, Figure
\ref{fig:robust_nonrobust_logMFA_EDD} presents another comparison of the performance between the robust CUSUM test and various non-robust CUSUM tests. The performance metrics chosen are the WADD defined in \eqref{eq:lorden} vs the log of the mean time to a false alarm (log MFA) defined in \eqref{eq:MFA}. 
It is well-known (\cite{veeravalli2014quickest}) that for any stopping time based on the CUSUM-like statistic,
\begin{equation}
    \label{eq:WADDworstcase1}
    \text{WADD}(\tau) = \mathsf{E}_1[\tau]. 
\end{equation}
Thus, the worst case is attained at $\nu=1$. Note that this is not the case for arbitrary stopping times. 
The values in the figure are obtained by choosing different values of the threshold \(A\) and estimating the delay by setting the change point \(\nu = 1\) and simulating the test for 5000 sample paths. For all tests, the post-change distribution chosen corresponds to ${J}=2$ (the worst case). The mean time to false alarm values is estimated by simulating the test for 5000 sample paths as well by setting the change point $\nu=\infty$ (no change).

\begin{figure}
    \centering
    \includegraphics[scale=0.8]
    {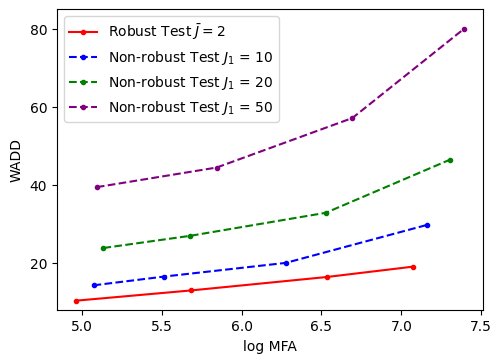}
    \caption{WADD vs. log MFA for robust and non-robust CUSUM tests. }
    \label{fig:robust_nonrobust_logMFA_EDD}
\end{figure}

The robust CUSUM test is represented by the solid red line. This test consistently exhibits lower WADD values across all log MFA ranges, indicating efficient detection performance with minimal delay. In contrast, the non-robust CUSUM tests are represented by the three dashed lines, each corresponding to a different post-change density assumption: \(J_1 = 10\) (blue), \(J_1 = 20\) (green), and \(J_1 = 50\) (purple). These non-robust tests show higher WADD values compared to the robust test, particularly as log MFA increases, which suggests slower and less efficient change detection.

\subsection{Comparison with a Non-parametric Test}
\label{sec:nonparametric}
In this section, we compare the optimal CUSUM test with a non-parametric test to show that the approximation provided by the asymptotic density \eqref{eq:PDF_V_delta1} is good enough to outperform a non-parametric test. The non-parametric test chosen 
is the standard non-parametric CUSUM test (\cite{brodsky1993nonparametric}): 
\[
W_m = \max\left\{ 0 , W_{m-1} + V_m - \frac{(m_0 + m_1)}{2} \right\} , \quad W_0 = 0
\]
where \( m_0 \) is the average of \( V \) before the change and \( m_1 \) is the average of \( V \) after the change. 

Figure \ref{fig:parametric_vs_nonparametric} compares the WADD between the parametric and non-parametric CUSUM algorithms as a function of $\log $(MFA). The parametric CUSUM algorithm uses \( J_1 = 3.62 \), which is the post-change sample estimated \( J \) value, rather than the designed robust test where ${\bar{J}} = 2$.  For the non-parametric CUSUM test, \( m_0 = 0.9117 \) and \( m_1 = 0.9467 \). The plot clearly shows that the non-parametric CUSUM algorithm has higher WADD values compared to the parametric CUSUM algorithm, indicating a delay in detection.

\begin{figure}
\centering
\includegraphics[scale=0.8] {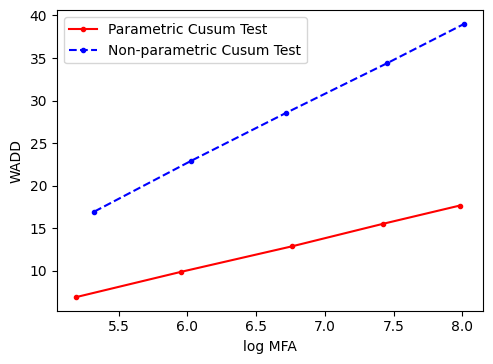}
\caption{WADD vs. log MFA for parametric and non-parametric CUSUM tests}
\label{fig:parametric_vs_nonparametric}
\end{figure}

\section{Conclusion}
In this paper, we have presented a novel robust algorithm for detecting changes in the correlation structure of high-dimensional data. The post-change correlation structure is allowed to be non-stationary (time-varying). We prove that the proposed algorithm is exactly robust optimal for detecting the non-stationary change, for every fixed constraint on the mean time to a false alarm (MFA). This is in contrast with many other tests in the literature that are only proved to be asymptotically optimal, as the MFA goes to infinity. 
We have also shown the effectiveness and robustness of our algorithm by applying it to simulated non-stationary data, and we have also shown that the proposed test outperforms non-parametric and non-robust tests.


\bibliographystyle{tfcad}
\bibliography{sources.bib, Taposh_QCD.bib}

\end{document}